\documentclass[12pt,leqno]{article}
\usepackage{amsfonts,amsthm,amsmath,lscape}
\newtheorem{thm}{Theorem}
\newtheorem{prop}[thm]{Proposition}

\theoremstyle{remark}
\newtheorem{rem}[thm]{Remark}

\newcommand{\FF}{\mathbb{F}}

\DeclareMathOperator{\wt}{wt}

\begin{document}

\title{Performance of ternary double circulant, double twistulant,
and self-dual codes
}

\author{
T. Aaron Gulliver\thanks{Department of Electrical and Computer Engineering,
University of Victoria,
P.O. Box 1700, STN CSC, Victoria, BC,
Canada V8W 2Y2.
email: agullive@ece.uvic.ca}
\text{ and }
Masaaki Harada\thanks{
Research Center for Pure and Applied Mathematics,
Graduate School of Information Sciences,
Tohoku University, Sendai 980--8579, Japan.
email: mharada@m.tohoku.ac.jp.}
}

\maketitle

\begin{abstract}
We study the performance of ternary isodual codes
which are not self-dual and ternary self-dual codes,
as measured by the decoding error probability in bounded distance decoding.
We compare the performance of ternary double circulant and double twistulant codes
which are not self-dual with ternary extremal self-dual codes.
We also investigate the performance of ternary self-dual codes
having large minimum weights.
\end{abstract}

\section{Introduction}
\label{Sec:Intro}
A (ternary) $[n,k]$ {\em code} $C$ is a $k$-dimensional vector subspace
of $\FF_3^n$,
where $\FF_3$ denotes the finite field of order $3$.
All codes in this paper are ternary.
We shall take the elements of $\FF_3$ to be either
\{$0,1,2$\} or \{$0,1,-1$\}, using whichever form is more convenient.
The parameter $n$ is called the {\em length} of $C$.
The {\em weight} $\wt(x)$ of a vector $x \in \FF_3^n$ is
the number of non-zero components of $x$.
A vector of $C$ is called a {\em codeword}.
The minimum non-zero weight of all codewords in $C$ is called
the {\em minimum weight} of $C$ and an $[n,k]$ code with minimum
weight $d$ is called an $[n,k,d]$ code.
Two codes $C$ and $C'$ are {\em equivalent} if there exists a
$(1,-1,0)$-monomial matrix $M$ with $C' =\{c M \mid c \in C \}$.

Let $C$ be an $[n,k,d]$ code.
Throughout this paper,
let $A_i$ denote the number of codewords of weight $i$ in $C$.
The  sequence $(A_0,A_1,\ldots,A_n)$ is called the {\em weight distribution} of $C$.
A code $C$ of length $n$ is said to be
{\em formally self-dual\/} if $C$ and $C^\perp$
have identical weight distributions where
$C^\perp$ is the dual code of $C$.
A code $C$ is {\em isodual\/} if $C$ and $C^\perp$ are equivalent, and
$C$ is {\em self-dual\/} if $C=C^\perp$.
It is known that a ternary self-dual code of length $n$ exists if and
only if $n \equiv 0 \pmod 4$.
A self-dual code is an isodual code, and an isodual code is a
formally self-dual code.
Double circulant and double twistulant codes
are a remarkable class of isodual codes.

The question of decoding error probabilities was studied
by Faldum, Lafuente, Ochoa and Willems~\cite{FLOW}
for bounded distance decoding.
Let $C$ and $C'$ be $[n,k,d]$ codes with weight distributions
$(A_0,A_1,\ldots,A_n)$ and $(A'_0,A'_1,\ldots,A'_n)$, respectively.
Suppose that symbol errors are independent and the symbol error probability is small.
Then $C$  has a smaller decoding error probability than $C'$
if and only if
\begin{equation}\label{eq:WD}
(A_0,A_1,\ldots,A_n) \prec (A'_0,A'_1,\ldots,A'_n),
\end{equation}
where $\prec$ means the lexicographic order, that is,
there is an integer $s \in \{0,1,\ldots,n\}$ such that
$A_i=A'_i$ for all $i < s$ but $A_s < A'_s$~\cite[Theorem~3.4]{FLOW}.
We say that $C$ {\em performs better} than $C'$ if~\eqref{eq:WD} holds.

In this paper, we investigate the performance of
optimal double circulant and double twistulant codes
which are not self-dual,
and self-dual codes with large minimum weights
as measured by the decoding error probability
in bounded distance decoding.
In Section~\ref{sec:DCC}, we compare the performance of double circulant
and double twistulant codes
which are not self-dual with extremal self-dual codes
for lengths $n < 48$.
Thus, we consider double circulant
and double twistulant codes which are not self-dual
only for lengths $n \equiv 0 \pmod 4$.
Self-dual codes are considered in Section~\ref{sec:SD}.
The weight distribution of an extremal self-dual code is
uniquely determined for each length.
For lengths up to $64$, the existence of an extremal self-dual code
is known (see~\cite[Table~4]{GG}).
The largest minimum weight of a self-dual code is $3\lfloor n/12 \rfloor$
for lengths $n=72,96$, and the largest minimum weight among currently known self-dual
codes is $3\lfloor n/12 \rfloor$
for lengths $n=68,76,80,84,88,92$ (see~\cite[Table~4]{GG}).
Hence, we investigate the performance of
self-dual codes of length $n$ and minimum weight
$3\lfloor n/12 \rfloor$ for $n=68,72,76,80,84,88,92,96$.

\section{Performance of double circulant and double twistulant codes}
\label{sec:DCC}

\subsection{Double circulant and double twistulant codes}

An $n \times n$ matrix is {\em circulant} or {\em negacirculant} if it has the form
\[
\left( \begin{array}{ccccc}
r_0     &r_1     & \cdots &r_{n-2}&r_{n-1} \\
cr_{n-1}&r_0     & \cdots &r_{n-3}&r_{n-2} \\
cr_{n-2}&cr_{n-1}& \ddots &r_{n-4}&r_{n-3} \\
\vdots  & \vdots &\ddots& \ddots & \vdots \\
cr_1    &cr_2    & \cdots&cr_{n-1}&r_0
\end{array}
\right),
\]
where $c=1$ or $-1$, respectively.
A {\em pure double circulant} code and a {\em bordered
double circulant} code have generator matrices of the form
\begin{equation}\label{eq:pDCC}
\left(\begin{array}{ccccc}
{} & I_n  & {} & R & {} \\
\end{array}\right)
\end{equation}
and
\begin{equation}\label{eq:bDCC}
\left(\begin{array}{ccccccccc}
{} & {} & {}      & {} & {} & \alpha & \beta  & \cdots & \beta  \\
{} & {} & {}      & {} & {} & \gamma     & {} & {}     &{} \\
{} & {} & I_{n} & {} & {} &\vdots & {} & R'     &{} \\
{} & {} & {}      & {} & {} & \gamma     & {} &{}      &{} \\
\end{array}\right),
\end{equation}
respectively,
where $I_n$ is the identity matrix of order $n$,
$R$ (resp.\ $R'$) is an $n \times n$ (resp.\ $n-1 \times n-1$)
circulant matrix, and $\alpha,\beta,\gamma \in \FF_3$.
These two families are called {\em double circulant} codes.

A classification of double circulant codes with the largest
minimum weight among all double circulant codes (including
self-dual codes)
was given in~\cite{DGH} for lengths up to $14$.
For lengths $n$ with $16 \le n \le 30$,
the largest minimum weight among all double circulant codes (including
self-dual codes) was determined in~\cite{DGH}.
For lengths $n$ with $16 \le n \le 24$,
the weight distributions for double circulant codes with the largest minimum weight
were determined and a double circulant code was given for each weight
distribution~\cite{DGH}.

A $[2n,n]$ code which has a generator matrix of the form
\begin{equation}\label{eq:QT}
\left(\begin{array}{cccccc}
{} & I_n  &  & & N & {} \\
\end{array}\right),
\end{equation}
where $N$ is an $n \times n$ negacirculant matrix, is called a {\em double twistulant} code.
Although the following proposition is somewhat trivial,
we provide it for the sake of completeness.

\begin{prop}
A double twistulant $[2n,n]$ code with generator matrix~\eqref{eq:QT} is isodual.
\end{prop}
\begin{proof}
The negacirculant matrix $N$ is obtained from $N^T$
by interchanging the $i$-th row (resp.\ column)
with the $(n+2-i)$-th row (resp.\ column)
($i=2,3,\ldots,\lfloor (n+1)/2\rfloor$),
and by negating the first row and column,
where $N^T$ denotes the transpose of $N$.
Hence, two codes with generator matrices
$
\left(\begin{array}{cc}
I_n  & N \\
\end{array}\right)
$
and
$
\left(\begin{array}{cc}
I_n  & N^T  \\
\end{array}\right)
$
are equivalent, and the result follows.
\end{proof}

In this section, we focus on double circulant codes and double twistulant codes
as a remarkable class of isodual codes.
We consider codes $C$ satisfying the following conditions:
\begin{itemize}
\item[(C1)]
$C$ is a pure or bordered double circulant codes or double twistulant codes of length
$n\ (\equiv 0 \pmod 4)$ which is not self-dual with the largest minimum weight $d_{P}$, $d_{B}$ and $d_{T}$
among pure or bordered double circulant codes or double twistulant codes
of length $n$ which are not self-dual, respectively.
\item[(C2)]
$C$ has the smallest weight distribution
$(A_0,A_1,\ldots,A_n)$ under the lexicographic order $\prec$
among pure or bordered double circulant codes, or
double twistulant codes of length $n$ and
minimum weight $d_{P}$, $d_{B}$ and $d_{T}$
which are not self-dual, respectively.
\end{itemize}

We say that a double circulant
(resp.\ double twistulant) code which is not self-dual
is {\em optimal} if it has the largest minimum weight among all
double circulant (resp.\ double twistulant)
codes of that length which are not self-dual.

\subsection{Performance of double circulant and double twistulant codes}

For $4m \le 48$, by determining the largest minimum weights $d_{P}$ (resp.\ $d_{B}$),
our exhaustive search found all distinct
pure (resp.\ bordered) double circulant $[4m,2m]$ codes
satisfying conditions (C1) and (C2).
This was done by considering all
$2m \times 2m$ circulant matrices $R$ in~\eqref{eq:pDCC}
(resp.\ $2m-1 \times 2m-1$ circulant matrices $R'$ in~\eqref{eq:bDCC}).
In addition, for $4m \le 48$,
by determining the largest minimum weights $d_{T}$,
our exhaustive search found all distinct double twistulant
$[4m,2m]$ codes satisfying conditions (C1) and (C2).
This was done by considering all
$2m \times 2m$ negacirculant matrices $N$ in~\eqref{eq:QT}.
Since a cyclic shift of the first row for a
code defines an equivalent code,
the elimination of cyclic shifts
substantially reduced the number of codes
which had to be checked further for equivalence to
complete the classification.
Then {\sc Magma}~\cite{Magma} was employed to determine code
equivalence which completed the classification for $4m < 48$.

In Table~\ref{Tab:Res}, we list the values
$d_{P}$, $A_{d_{P}}$,
$d_{B}$, $A_{d_{B}}$,
$d_{T}$ and $A_{d_{T}}$.
We also list the inequivalent pure and bordered double circulant codes, and double twistulant codes
satisfying conditions (C1) and (C2).
For the codes listed in the table,
the first rows of $R$ in~\eqref{eq:pDCC},
$R'$ in~\eqref{eq:bDCC} and
$N$ in~\eqref{eq:QT}
are given in Tables~\ref{Tab:P}, \ref{Tab:B} and \ref{Tab:Q}, 
respectively.
The border values $(\alpha,\beta,\gamma)$ in~\eqref{eq:bDCC}
are listed for the bordered double circulant codes.
In addition, 
the minimum weights $d$ and $(A_d,A_{d+1},A_{d+2})$ are listed.
As mentioned above, for lengths $n$ with $16 \le n \le 24$,
all weight distributions for
double circulant codes with the largest minimum weight
were determined (\cite[Table~6]{DGH})
and a pure and bordered double circulant code was
given for each weight distribution (\cite[Tables~4 and 5]{DGH}).
The bordered double circulant code $B_{16}$ in Table~\ref{Tab:B}
has the following weight distribution
\begin{multline*}
A_0=1,
A_6=  84,
A_7= 336,
A_8= 420,
A_9= 872,
A_{10}=1092,
\\
A_{11}=1680,
A_{12}= 924,
A_{13}= 840,
A_{14}= 168,
A_{15}= 144.
\end{multline*}
Since this weight distribution was not given in~\cite[Table~6]{DGH},
the code $B_{16}$ should be added to~\cite[Table~5]{DGH}.


To compare the performance of the optimal double circulant and double twistulant codes
which are not self-dual as measured by the decoding error probability with
bounded distance decoding, we list the largest minimum weight $d_{SD}$
and the smallest number $A_{SD}$ of codewords of weight $d_{SD}$ among self-dual codes of length $n$.
It was shown in~\cite{MS73} that the minimum weight $d$ of a self-dual code of length $n$ is bounded by
$d\leq 3 \lfloor n/12 \rfloor +3$.
If $d=3 \lfloor n/12 \rfloor+3$, then the code is called {\em extremal}.
The weight distribution of an extremal self-dual code of length $n$ is uniquely determined (see~\eqref{eq:WE}).
For lengths up to $64$, the existence of an extremal self-dual code is known (see~\cite[Table~4]{GG}).
Hence, $d_{SD}$ and $A_{SD}$ in Table~\ref{Tab:Res} are uniquely determined for each length.

For the cases $d_P > d_{SD}$, $d_B > d_{SD}$ and
$d_T > d_{SD}$\footnote{These cases are marked by $*$ in
columns $d_P$, $d_B$ and $d_T$ of Table~\ref{Tab:Res}.},
we list in Table~\ref{Tab:Res}
the inequivalent pure and bordered
double circulant codes, and double twistulant codes
$C$ with minimum weight $d_{SD}$ satisfying that
$C$ has the smallest weight distribution
among pure and bordered double circulant codes, and
double twistulant codes of length $n$ and
minimum weight $d_{SD}$ which are not self-dual.

From Table~\ref{Tab:Res}, we have the following results concerning the
performance of double circulant and double twistulant codes which are not self-dual.

\begin{thm}\label{thm}
Suppose that
\begin{align*}
(n,d) =& (8,3), (16,6), (20,6), (28,9), (32,9), (44,12).
\end{align*}
Then there is a double circulant
$[n,n/2,d]$ code $C$ which is not self-dual
and a double twistulant
$[n,n/2,d]$ code $C$ which is not self-dual
such that $C$ performs better than any self-dual $[n,n/2,d]$ code.
\end{thm}

\begin{rem}
For $n=8,20,32,44$ (resp.\ $n=8,20,44$),
there is a double circulant (resp.\ double twistulant) code
$C$ of length $n$ which is not self-dual such that
$C$ has a larger minimum weight than any self-dual code of length $n$.
\end{rem}

\begin{rem}
For length $48$, we verified that $d_P=12, d_B=14, d_T=12$.
Also we verified that there are three inequivalent
bordered double circulant $[48,24,14]$ codes $B_{48,i}$ $(i=1,2,3)$
which are not self-dual.  
For the three codes, the first rows of $R'$ and
the border values $(\alpha,\beta,\gamma)$ in~\eqref{eq:bDCC}
are also given in Table~\ref{Tab:B}.
\end{rem}

\begin{table}[thbp]
\caption{Pure double circulant codes}
\label{Tab:P}
\begin{center}
{\footnotesize
\begin{tabular}{c|r|c|c}
\noalign{\hrule height1pt}
Code &  \multicolumn{1}{c|}{First row} & $d$ & $(A_d,A_{d+1},A_{d+2})$\\
\hline
$P_{4,1}$ & (11)& 2& $(2, 4, 2)$\\
$P_{4,2}$ & (12)& 2& $(2, 4, 2)$\\
$P_{12,1}$ & (122010)& 5& $(48, 98, 132 )$\\
$P_{12,2}$ & (111201)& 5& $(48, 98, 132 )$\\
$P_{24}$ & (110020021021)& 8& $(348, 1776, 3912)$\\
$P_{28}$ & (11100121001121) & 9& $(924,3220,9996)$ \\
$P_{32,1}$ & (1201110000101101) &10& $(2208,8832,7728)$ \\
$P_{32,2}$ & (1221022111112212) &10& $(2208,8832,7728)$ \\
$P_{36}$ & (112202210120102222) &10& $(270,3636,15042)$ \\
$P_{40}$ & (10122212102112102100) &11& $(720,8120,29440)$ \\
$P_{44,1}$ & (2021101121112021101000) &13& $(19712,87296, 87296)$ \\
$P_{44,2}$ & (1112101101011001100000) &13& $(19712,87296, 87296)$ \\
\hline
$P'_{8}$ & (1100) &6& $(8,10,16)$ \\
$P'_{20}$ & (1200112220) &6& $(10, 180, 680)$ \\
$P'_{32}$ & (1021022000211011) &9& $(64,1600,7616)$ \\
$P'_{44}$ & (1211112022021010110000) &12& $(1716,15752,65120)$ \\
\noalign{\hrule height1pt}
\end{tabular}
}
\end{center}
\end{table}

\begin{table}[thbp]
\caption{Bordered double circulant codes}
\label{Tab:B}
\begin{center}
{\footnotesize
\begin{tabular}{c|r|c|c|c}
\noalign{\hrule height1pt}
Code &  \multicolumn{1}{c|}{First row} &$(\alpha,\beta,\gamma)$
& $d$ & $(A_d,A_{d+1},A_{d+2})$\\
\hline
$B_{ 4}$ & (2)     &$(0,1,1)$ & 4 & $(2,4,2)$\\
$B_{ 8}$ & (112)   &$(0,1,1)$ & 4 & $(22,24,20)$\\
$B_{12}$ & (22101) &$(1,2,2)$ & 5 & $(30, 162, 72)$\\
$B_{16}$&(2110100) &$(1,1,1)$ & 6 & $(84,336,420)$ \\
$B_{24,1}$ &(11102122021)& $(2,2,2)$ & 8 &$(264,2794,990)$\\
$B_{24,2}$ &(11121021222)& $(2,2,2)$ & 8 &$(264,2794,990)$\\
$B_{24,3}$ &(21200112221)& $(0,2,2)$ & 8 &$(264,2794,990)$\\
$B_{28}$ & (1102202200222)&$(1,1,1)$ & 9 &$(832,3536,9880)$\\
$B_{32}$ & (222011121020010)&$(0,2,2)$ & 9 &$(60,1870,6876)$\\
$B_{36}$ & (11000101121101000) &$(0,1,1)$  & 11 &$(2244,30804,9792)$\\
$B_{40}$ & (1012021012200110000) &$(1,1,1)$  & 11 &$(722,7790,31084)$\\
$B_{44}$& (120111201121200101100) &$(0,1,1)$ & 12 &$(2436,15470,61278)$\\
$B_{48,1}$& (11202002011021101001000) &$(0,1,1)$ & 14 &$(19320, 304704, 91080)$\\
$B_{48,2}$& (21111010110011001010000) &$(0,1,1)$ & 14 &$(19320, 304704, 91080)$\\
$B_{48,3}$& (12011112120211110001000) &$(0,1,1)$ & 14 &$(19320, 304704, 91080)$\\
\hline
$B'_8$ & (102)&$(2,2,2)$ & 3 &$(2,18,26)$\\
$B'_{20}$ & (112021000)&$(0,2,2)$ & 6 &$(6, 216, 594)$\\
\noalign{\hrule height1pt}
\end{tabular}
}
\end{center}
\end{table}

\begin{table}[thbp]
\caption{Double twistulant codes}
\label{Tab:Q}
\begin{center}
{\footnotesize
\begin{tabular}{c|r|c|c}
\noalign{\hrule height1pt}
Code &  \multicolumn{1}{c|}{First row}
& $d$ & $(A_d,A_{d+1},A_{d+2})$\\
\hline
$T_{4 }$ & (10)     & 2 & $(4,0,4)$\\
$T_{8 }$ & (1201)     & 4 & $(24,16,32)$\\
$T_{12}$ & (121010)     & 5 & $(8,96,144)$\\
$T_{16}$ & (10021102)     & 6 & $(96,288,496)$\\
$T_{20,1}$ & (1101001011)     & 7 & $(200,680,1560)$\\
$T_{20,2}$ & (1012220001)     & 7 & $(200,680,1560)$\\
$T_{20,3}$ & (1110020021)     & 7 & $(200,680,1560)$\\
$T_{24}$ & (111221210220)     & 8 & $(312,1928,3696)$\\
$T_{28}$ & (12211210012220)   & 9 & $(616,4200,9632)$\\
$T_{32}$ & (1011122111110100)   & 9 & $(32,1856,7360)$\\
$T_{36}$ & (222120202011110000)   & 10 & $(252,3816, 14868)$\\
$T_{40}$ & (12020112202100100000)   & 11 & $(480,10800,24160)$\\
$T_{44}$ & (1120222200211010001000)   & 13 & $(19712,87296,87296)$\\
\hline
$T'_{8,1}$ & (1120)     & 3 &$(8,8,24)$\\
$T'_{8,2}$ & (1122)     & 3 &$(8,8,24)$\\
$T'_{20}$ & (1201221021)     & 6 &$(20,140,780)$\\
$T'_{44}$ & (1112121111221221110100) & 12 &$(1716, 15752, 65120)$\\
\noalign{\hrule height1pt}
\end{tabular}
}
\end{center}
\end{table}

\section{Performance of self-dual codes}
\label{sec:SD}

In this section, we investigate the performance of
self-dual codes of length $n$ and minimum weight $3\lfloor n/12 \rfloor$
for $n=68,72,76,80,84,88,92,96$.

\subsection{Largest minimum weights}

As mentioned above, the minimum weight
$d$ of a self-dual code of length $n$ is bounded by
$d\leq 3 \lfloor n/12 \rfloor +3$~\cite{MS73}, and
a self-dual code with $d=3 \lfloor n/12 \rfloor +3$
is called {\em extremal}.
We say that a self-dual code of length $n$
is {\em optimal} if it has the largest minimum weight among all
self-dual codes of that length.
Of course, an extremal self-dual code is optimal.

The weight enumerator of a code of length $n$ is
defined as $\sum_{i=0}^n A_i y^i$.
The weight enumerator $W$ of a self-dual code of length $n$
can be represented as an integral combination of Gleason polynomials
(see~\cite{MS73}), so that
\begin{equation}\label{eq:WE}
W= \sum_{j=0}^{\lfloor n/12 \rfloor}a_j(1+y^3)^{n/4-3j}(y^3(1-y^3)^3)^j,
\end{equation}
for some integers $a_j$ with $a_0=1$.
Since the weight enumerator of an extremal self-dual code
of length $n$ is uniquely determined, all extremal self-dual
codes of length $n$ have the same performance.
Note that the weight enumerator of a self-dual
$[n,n/2,3\lfloor n/12 \rfloor]$ code
can be expressed using a single integer variable.

For lengths up to $64$, the existence of an extremal self-dual code
is known (see~\cite[Table~4]{GG}).
It is also known that there is no extremal self-dual
code for lengths $72$ and $96$,
and that there are self-dual codes with parameters
$[72,36,18]$ and $[96,48,24]$.
In addition, the largest minimum weight among currently known self-dual
codes of length $n$ is $3\lfloor n/12 \rfloor$
for $n=68,76,80,84,88,92$.
This is a reason for investigating the performance of
self-dual codes of length $n$ and minimum weight $3\lfloor n/12 \rfloor$
for $n=68,72,76,80,84,88,92,96$.

\subsection{Methods for constructing self-dual codes}

Many self-dual codes with large minimum weights
were constructed as double circulant codes
and double twistulant codes~\cite{GG}.
In this section, the following two methods for
constructing self-dual codes are employed.

Let $C$ and $D$ be self-dual codes of lengths $m$ and $n$,
respectively, where $n/2 \ge m$.
Let $d_1$ be the maximum weight among the codewords of $C$,
and let $d$ be the minimum weight of $D$.
Suppose that the set of the first $m$ coordinates $\Gamma_D$ of $D$
is a subset of some information set.
Let $E$ be the code consisting of all vectors
$y \in \FF_3^{n-m}$ such that $(x,y) \in D$ for some $x \in C$.
Then $E$ is a self-dual $[n-m,(n-m)/2]$ code with
minimum weight at least $d-d_1$~\cite{CPS79}.
By considering the other $m$ coordinates $\Gamma_D$,
many self-dual codes can be constructed.
We say that these self-dual codes of length $n-m$ are
constructed from $D$ by subtracting $C$.
In this section, we consider self-dual codes of length $n-4$
from a self-dual $[n,n/2,3\lfloor n/12 \rfloor]$ code
by subtracting the unique self-dual $[4,2,3]$ code $e_4$
for $n=72,84,96$.

A {\em four-negacirculant} $[4n,2n]$ code has
a generator matrix of the form
\begin{equation}\label{eq:G}
\left(
\begin{array}{ccc@{}c}
\quad & {\Large I_{2n}} & \quad &
\begin{array}{cc}
A & B \\
-B^T & A^T
\end{array}
\end{array}
\right),
\end{equation}
where $A$ and $B$ are negacirculant matrices.
Many extremal self-dual codes are four-negacirculant codes~\cite{HHKK}.
In this section, we use this construction to obtain self-dual codes with
minimum weight $3\lfloor n/12 \rfloor$ for $n=68,72,76,80,92$.

\subsection{Self-dual $[68,34,15]$ codes}
Using~\eqref{eq:WE}, the weight enumerator of a self-dual $[68,34,15]$ code is
\begin{align*}
&
1 + a y^{15}
+ (596904 + a) y^{18}
+ (70982208 - 71 a) y^{21}
\\ &
+ (4537453680 + 265 a) y^{24}
+ (164380156864 + 805 a) y^{27}
\\ &
+ (3452859764640 - 10283 a) y^{30}
+ \cdots
+ (30394368 - 64 a)y^{66},
\end{align*}
where $a$ is an integer with $1 \le a \le 474912$.
A self-dual $[68,34,15]$ code can be constructed from
the extended quadratic residue code of length $72$
by subtracting $e_4$.
In this way, we found self-dual $[68,34,15]$ codes with weight distributions where
\begin{align*}
A_{15} =&
3592, 3624, 3628, 3632, 3652, 3696, 3708, 3712, 3728, 3732, 3736, 3768,
\\ &
3772, 3776, 3796, 3808, 3840, 3852, 3856, 3872, 3876, 3912, 3916, 3920,
\\ &
3940, 3952, 3984, 3996, 4000, 4016, 4020, 4064, 4168, 4200, 4208, 4272,
\\ &
4304, 4384.
\end{align*}
A self-dual $[68,34,15]$ code can be found in~\cite[Table~3]{GO}.
A double twistulant self-dual $[68,34,15]$ code can be found in~\cite[Table~3]{GG}.
We verified by {\sc Magma}~\cite{Magma} that these codes have $A_{15}=1224$ and $3128$, respectively.

\begin{table}[thb]
\caption{Four-negacirculant self-dual $[68,34,15]$ codes}
\label{Tab:68}
\begin{center}
{\footnotesize
\begin{tabular}{c|c|c}
\noalign{\hrule height1pt}
Code & $(r_A,r_B)$ & $A_{15}$ \\
\hline
$C_{68, 1}$&$((12211110002000221),(11020202120022121))$ & 1088 \\
$C_{68, 2}$&$((02100210211122010),(22122110221200002))$ & 1428 \\
$C_{68, 3}$&$((02100210211122010),(11112122100102002))$ & 1496 \\
$C_{68, 4}$&$((02100210211122010),(12222021011101002))$ & 1564 \\
$C_{68, 5}$&$((02100210211122010),(12102010210100002))$ & 1632 \\
$C_{68, 6}$&$((02100210211122010),(10120210101200002))$ & 1700 \\
$C_{68, 7}$&$((02100210211122010),(12011201122110002))$ & 1768 \\
$C_{68, 8}$&$((02100210211122010),(01202021220200002))$ & 1836 \\
$C_{68, 9}$&$((00102220000220100),(12121100111010122))$ & 1904 \\
$C_{68,10}$&$((12211110002000221),(10022100222022121))$ & 1972 \\
$C_{68,11}$&$((02100210211122010),(10211100201200002))$ & 2040 \\
$C_{68,12}$&$((02100210211122010),(11110200221000002))$ & 2244 \\
$C_{68,13}$&$((02100210211122010),(01110022120100002))$ & 2312 \\
$C_{68,14}$&$((00102220000220100),(12112010121010122))$ & 2380 \\
$C_{68,15}$&$((02100210211122010),(01001002200100002))$ & 2516 \\
$C_{68,16}$&$((00102220000220100),(10110110001010122))$ & 2584 \\
$C_{68,17}$&$((00102220000220100),(12112210011010122))$ & 2652 \\
$C_{68,18}$&$((12211110002000221),(02202210102022121))$ & 2856 \\
$C_{68,19}$&$((12211110002000221),(00022021122022121))$ & 3196 \\
$C_{68,20}$&$((00102220000220100),(01111202221010122))$ & 3468 \\
\noalign{\hrule height1pt}
\end{tabular}
}
\end{center}
\end{table}

By considering four-negacirculant codes, we found new
self-dual $[68,34,15]$ codes $C_{68,i}$ $(i=1,2,\ldots,20)$.
The first rows $r_A$ and $r_B$
of negacirculant matrices $A$ and $B$ in~\eqref{eq:G} are
listed in Table~\ref{Tab:68}.
The numbers $A_{15}$ for these codes are also listed in the table.
Hence, $C_{68,1}$ performs better than the above self-dual codes constructed by
subtraction, the two codes in~\cite[Table~3]{GO}, \cite[Table~3]{GG}
and $C_{68,i}$ $(i=2,3,\ldots,20)$.

\subsection{Optimal self-dual $[72,36,18]$ codes}
Using~\eqref{eq:WE}, the weight enumerator of an optimal self-dual $[72,36,18]$ code is
\begin{align*}
&
1 + a y^{18}
+ (36213408 - 18 a) y^{21}
+ (2634060240 + 153 a) y^{24}
\\ &
+ (126284566912 - 816 a) y^{27}
+ (3525613242624 + 3060 a) y^{30}
\\ &
+ (59358705673680 - 8568 a) y^{33}
+ \cdots
+ (- 115728 + a)y^{72},
\end{align*}
where $a$ is an integer with $115728 \le a \le 2011856$.
The extended quadratic residue code of length $72$ is a self-dual
code with $d=18$ and $A_{18}=357840$~\cite{GNW}.
A double twistulant self-dual $[72,36,18]$ code can be found in~\cite[Table~3]{GG}.
We verified by {\sc Magma}~\cite{Magma} that this code has $A_{18}=213936$.

\begin{table}[thb]
\caption{Four-negacirculant self-dual $[72,36,18]$ codes}
\label{Tab:72}
\begin{center}
{\footnotesize
\begin{tabular}{c|c|c}
\noalign{\hrule height1pt}
Code & $(r_A,r_B)$ & $A_{18}$ \\
\hline
$C_{72,1}$&$((012102100100020021),(022101002101002112))$&205464\\
$C_{72,2}$&$((210121111222022000),(200222220120220212))$&209184\\
$C_{72,3}$&$((111022210021122000),(001101111111000012))$&209736\\
$C_{72,4}$&$((001121111020012112),(221221012112221110))$&210456\\
$C_{72,5}$&$((100111020012120220),(020200102102020022))$&212280\\
$C_{72,6}$&$((002020000222220002),(102201201102022210))$&212376\\
$C_{72,7}$&$((201201222122110010),(120012101020201100))$&213456\\
$C_{72,8}$&$((010201112111021012),(022222102122021200))$&213648\\
$C_{72,9}$&$((012201211210110112),(220022210111120001))$&213744\\
$C_{72,10}$&$((010202021120102002),(021020222222020112))$&214992\\
\noalign{\hrule height1pt}
\end{tabular}
}
\end{center}
\end{table}

By considering four-negacirculant codes, we found new
self-dual $[72,36,18]$ codes $C_{72,i}$ $(i=1,2,\ldots,10)$.
The first rows $r_A$ and $r_B$ of the negacirculant matrices $A$ and $B$ in~\eqref{eq:G} are
listed in Table~\ref{Tab:72}.
The numbers $A_{18}$ for these codes are also listed in the table.
Hence, $C_{72,1}$ performs better than
the two previously known codes and $C_{72,i}$ $(i=2,3,\ldots,10)$.

\subsection{Self-dual $[76,38,18]$ codes}
Using~\eqref{eq:WE}, the weight enumerator of a self-dual $[76,38,18]$ code is
\begin{align*}
&
1 + a y^{18}
+ (14228720 - 10 a) y^{21}
+ (1403328600 + 9 a) y^{24}
\\ &
+ (84823417600 + 408 a) y^{27}
+ (3080650381440 - 3468 a) y^{30}
\\ &
+ (68562946755000 + 15912 a) y^{33}
+ \cdots
+ (5820992 + 8 a) y^{75},
\end{align*}
where $a$ is an integer with $1 \le a \le 1422872$.
A self-dual $[76,38,18]$ code can be found in~\cite[Table~3]{GO}, and
a double twistulant self-dual $[76,38,18]$ code can be found in~\cite[Table~3]{GG}.
We verified by {\sc Magma}~\cite{Magma}
that these codes have $A_{18}=71136$ and $75088$, respectively.

\begin{table}[thb]
\caption{Four-negacirculant self-dual $[76,38,18]$ codes}
\label{Tab:76}
\begin{center}
{\footnotesize
\begin{tabular}{c|c|c}
\noalign{\hrule height1pt}
Code & $(r_A,r_B)$ & $A_{18}$ \\
\hline
$C_{76, 1}$ & $((1201011221200201200),(1110000120120012120))$ & 65436 \\
$C_{76, 2}$ & $((1112120212012212011),(1220022211120100012))$ & 65968 \\
$C_{76, 3}$ & $((1012111100001222101),(1202111220002221200))$ & 67868 \\
$C_{76, 4}$ & $((1012000220122101201),(1021212120111210100))$ & 68096 \\
$C_{76, 5}$ & $((2112210110201220011),(1121212202120112010))$ & 68628 \\
$C_{76, 6}$ & $((1012221221202111020),(2000000221110021000))$ & 68704 \\
$C_{76, 7}$ & $((1012000220122101201),(0022221101202200100))$ & 69844 \\
$C_{76, 8}$ & $((1112120212012212011),(2020102100210100012))$ & 69996 \\
$C_{76, 9}$ & $((1012000220122101201),(1110001022221200100))$ & 70376 \\
$C_{76,10}$ & $((1201011221200201200),(0201020102110012120))$ & 70452 \\
$C_{76,11}$ & $((1012111100001222101),(2212020011021121200))$ & 70604 \\
$C_{76,12}$ & $((1012000220122101201),(0101221120212000100))$ & 70832 \\
$C_{76,13}$ & $((1201011221200201200),(2211220201001112120))$ & 70908 \\
$C_{76,14}$ & $((1012000220122101201),(0000220211200200100))$ & 70984 \\
$C_{76,15}$ & $((1201011221200201200),(1111020012022212120))$ & 71212 \\
$C_{76,16}$ & $((2112210110201220011),(1220020022111012010))$ & 71364 \\
$C_{76,17}$ & $((1201011221200201200),(0101120021112212120))$ & 71668 \\
$C_{76,18}$ & $((1112120212012212011),(0122201111011100012))$ & 71744 \\
$C_{76,19}$ & $((1012111100001222101),(2212220200102121200))$ & 72200 \\
$C_{76,20}$ & $((2112210110201220011),(1121010100211012010))$ & 73340 \\
\noalign{\hrule height1pt}
\end{tabular}
}
\end{center}
\end{table}

By considering four-negacirculant codes, we found new
self-dual $[76,38,18]$ codes $C_{78,i}$ $(i=1,2,\ldots,20)$.
The first rows $r_A$ and $r_B$
of negacirculant matrices $A$ and $B$ in~\eqref{eq:G} are
listed in Table~\ref{Tab:76}.
The numbers $A_{18}$ for these codes are also listed in the table.
Hence, $C_{76,1}$ performs better than
the two previously known codes and $C_{76,i}$ $(i=2,3,\ldots,20)$.

\subsection{Self-dual $[80,40,18]$ codes}
Using~\eqref{eq:WE}, the weight enumerator of a self-dual $[80,40,18]$ code is
\begin{align*}
&
1 + a y^{18}
+ (5262400 - 2 a) y^{21}
+ (673223200 - 71 a) y^{24}
\\ &
+ (50911463680 + 480 a) y^{27}
+ (2350521997824 - 204 a) y^{30}
\\ &
+ (67551815604000 - 11832 a) y^{33}
+ \cdots
+ (234280960 + 64 a) y^{78},
\end{align*}
where $a$ is an integer with $1 \le a \le 2631200$.
A self-dual $[80,40,18]$ code can be constructed from
the extended quadratic residue code 
of length $84$ by subtracting $e_4$.
In this way, we found self-dual $[80,40,18]$ codes with
weight distributions where
\begin{align*}
A_{18} =&
25400, 25444, 25488, 25508, 25528, 25544, 25552, 25576, 25592,
\\&
25596, 25604, 25616, 25636, 25652, 25660, 25672, 25676, 25684,
\\&
25688, 25700, 25712, 25720, 25724, 25744, 25764, 25780, 25784,
\\&
25792, 25796, 25808, 25820, 25828, 25832, 25852, 25872, 25888,
\\&
25892, 25904, 25920, 25960, 25976, 25980, 26152, 26176.
\end{align*}
A self-dual $[80,40,18]$ code can be found in~\cite[Table~3]{GO}.
A double twistulant self-dual $[80,40,18]$ code can be found
in~\cite[Table~3]{GG}.
We verified by {\sc Magma}~\cite{Magma}
that these codes have $A_{18}=21320$ and $20960$, respectively.

\begin{table}[thb]
\caption{Four-negacirculant self-dual $[80,40,18]$ codes}
\label{Tab:80}
\begin{center}
{\footnotesize
\begin{tabular}{c|c|c}
\noalign{\hrule height1pt}
Code & $(r_A,r_B)$ & $A_{18}$ \\
\hline
$C_{80, 1}$&$((00210000002212002101),(22221120102012011001))$& 19360\\
$C_{80, 2}$&$((11020012111021001111),(10021112112011011202))$& 19760\\
$C_{80, 3}$&$((11020012111021001111),(11112121010011011202))$& 20240\\
$C_{80, 4}$&$((10112222112121111221),(20112102021100112001))$& 20480\\
$C_{80, 5}$&$((10100222101020022100),(20102101121021021112))$& 20640\\
$C_{80, 6}$&$((12121110011110210020),(02000020121112202210))$& 20720\\
$C_{80, 7}$&$((21000220020020220200),(01210202212020121210))$& 21120\\
$C_{80, 8}$&$((02202212000212102221),(21200212221202022011))$& 21200\\
$C_{80, 9}$&$((11211221011111022222),(01010012201212111120))$& 21600\\
$C_{80,10}$&$((11020012111021001111),(00122020002111011202))$& 21920\\
$C_{80,11}$&$((01200012212101200110),(00221200012110011200))$& 22160\\
$C_{80,12}$&$((02202212000212102221),(01210221112202022011))$& 22320\\
$C_{80,13}$&$((02011101122110012111),(00020202221101122121))$& 22720\\
$C_{80,14}$&$((01200012212101200110),(21210001022010011200))$& 22960\\
$C_{80,15}$&$((02011101122110012111),(20012002122001122121))$& 23280\\
$C_{80,16}$&$((21000220020020220200),(20000221020020121210))$& 23680\\
$C_{80,17}$&$((12121110011110210020),(22022222022102202210))$& 24160\\
$C_{80,18}$&$((12121110011110210020),(02002021202202202210))$& 24320\\
$C_{80,19}$&$((12121110011110210020),(21212110011012202210))$& 24400\\
$C_{80,20}$&$((10112222112121111221),(01022121011100112001))$& 24800\\
\noalign{\hrule height1pt}
\end{tabular}
}
\end{center}
\end{table}

By considering four-negacirculant codes, we found new
self-dual $[80,40,18]$ codes $C_{80,i}$ $(i=1,2,\ldots,20)$.
The first rows $r_A$ and $r_B$
of negacirculant matrices $A$ and $B$ in~\eqref{eq:G} are
listed in Table~\ref{Tab:80}.
The numbers $A_{18}$ for these codes are also listed in the table.
Hence, $C_{80,1}$ performs better than
the above self-dual codes constructed by subtraction,
the two codes in~\cite[Table~3]{GO}, \cite[Table~3]{GG}
and $C_{80,i}$ $(i=2,3,\ldots,20)$.

\subsection{Self-dual $[84,42,21]$ codes}
Using~\eqref{eq:WE}, the weight enumerator of a self-dual $[84,42,21]$ code is
\begin{align*}
&
1 + a y^{21}
+ (128391120 - 13 a) y^{24}
+ (13697686464 + 42 a) y^{27}
\\ &
+ (972882111168 + 350 a) y^{30}
+ (44029165524624 - 4655 a) y^{33}
\\ &
+ (1294136458420608 + 27531 a) y^{36}
+ \cdots
+ (46354176 - 8 a) y^{87},
\end{align*}
where $a$ is an integer with $1 \le a \le 5794272$.
The extended quadratic residue code $QR_{84}$ and
the Pless symmetry code $P_{84}$ of length $84$ are currently
the only known self-dual $[84,42,21]$ codes.
The code $QR_{84}$ has $A_{21}=2368488$ and $P_{84}$ has $A_{21}=1259520$~\cite{GNW}.
This means that $P_{84}$ performs better than $QR_{84}$.
Our extensive search failed to discover a four-negacirculant
self-dual $[84,42,21]$ code.

\subsection{Self-dual $[88,44,21]$ codes}
Using~\eqref{eq:WE}, the weight enumerator of a self-dual $[88,44,21]$ code is
\begin{align*}
&
1 + a y^{21}
+ (128391120 - 13 a) y^{24}
+ (13697686464 + 42 a) y^{27}
\\ &
+ (972882111168 + 350 a) y^{30}
+ (44029165524624 - 4655 a) y^{33}
\\ &
+ (1294136458420608 + 27531 a) y^{36}
+ (25036311539416320 - 108528 a) y^{39}
\\ &
+ \cdots
+ (46354176 - 8 a) y^{87},
\end{align*}
where $a$ is an integer with $1 \le a \le 5794272$.
A self-dual $[88,44,21]$ code can be found
in~\cite{Be84} (see also~\cite[Table~4]{GO}).
This code has $A_{21}=635712$~\cite{GNW}.
Our extensive search failed to discover a four-negacirculant
self-dual $[88,44,21]$ code.

\subsection{Self-dual $[92,46,21]$ codes}
Using~\eqref{eq:WE}, the weight enumerator of a self-dual $[92,46,21]$ code is
\begin{align*}
&
1 + a y^{21}
+ (46823400 - 5 a) y^{24}
+ (6304654752 - 62 a) y^{27}
\\ &
+ (541436863968 + 686 a) y^{30}
+ (30032673751080 - 1855 a) y^{33}
\\ &
+ (1093919194221984 - 9709 a) y^{36}
+ (26544142192296960 + 111720 a) y^{39}
\\ &
+ \cdots
+ (1766329344 - 64 a) y^{90},
\end{align*}
where $a$ is an integer with $1 \le a \le 9364680$.
We found ten self-dual $[92,46,21]$ codes constructed from
the Pless symmetry code $P_{96}$
of length $96$ by subtracting $e_4$, and
we determined that these codes have weight distributions
where
\begin{align*}
A_{21} =&
170536, 171300, 171772, 172000, 172344, 172392, 172640, 172668,
\\&
172764, 173236.
\end{align*}
A double twistulant self-dual $[92,46,21]$ code can be found in~\cite[Table~3]{GG}.
We verified by {\sc Magma}~\cite{Magma} that this code has $A_{21}=204608$.

\begin{table}[thb]
\caption{Four-negacirculant self-dual $[92,46,21]$ codes}
\label{Tab:92}
\begin{center}
{\footnotesize
\begin{tabular}{c|c|c}
\noalign{\hrule height1pt}
Code & $(r_A,r_B)$ & $A_{21}$ \\
\hline
$C_{92,1}$&$((00022211101222011012202),(12022121221012000210110))$&190532\\
$C_{92,2}$&$((10011201111102210102101),(20020102221000122101100))$&192648\\
$C_{92,3}$&$((10011201111102210102101),(00111022122201221101100))$&195408\\
$C_{92,4}$&$((00022211101222011012202),(10101220110122210210110))$&196696\\
$C_{92,5}$&$((21022000211022222102200),(22112012202002220122021))$&197892\\
$C_{92,6}$&$((10011201111102210102101),(11101212111120010201100))$&199916\\
$C_{92,7}$&$((00022211101222011012202),(20102001112021211210110))$&201388\\
$C_{92,8}$&$((21022000211022222102200),(11112220002020112222021))$&201572\\
$C_{92,9}$&$((21022000211022222102200),(20001011100020212122021))$&202676\\
$C_{92,10}$&$((00022211101222011012202),(20220022012010020210110))$&203688\\
\noalign{\hrule height1pt}
\end{tabular}
}
\end{center}
\end{table}

By considering four-negacirculant codes, we found new
self-dual $[92,46,21]$ codes $C_{92,i}$ $(i=1,2,\ldots,10)$.
The first rows $r_A$ and $r_B$
of negacirculant matrices $A$ and $B$ in~\eqref{eq:G} are
listed in Table~\ref{Tab:92}.
The numbers $A_{21}$ for these codes are also listed in the table.
Hence, the code $N_{92}$ constructed from $P_{96}$ by
subtracting $e_4$ with $A_{21}=170536$
performs better than
the code in~\cite[Table~3]{GG},
the nine other codes constructed from $P_{96}$ by
subtracting $e_4$, and $C_{92,i}$ $(i=1,2,\ldots,10)$.

To define $N_{92}$, we give the generator matrices of $P_{96}$ and $e_4$.
The code $P_{96}$ is the bordered double circulant code with the first row of
$R'$ in~\eqref{eq:bDCC} given by
\[
(01111211112212121112212211211222121211222212222),
\]
and border values $(\alpha,\beta,\gamma)=(0,1,1)$.
The code $e_4$ has generator matrix
${\displaystyle
\left( \begin{array}{ccccc}
1&0&1&1\\
0&1&1&2
\end{array}
\right)
}$.
Then $N_{92}$ is constructed from $P_{96}$ by subtracting $e_4$ where
the four coordinates are $\Gamma_{P_{96}}=(1,2,3,16)$.


\subsection{Optimal self-dual $[96,48,24]$ codes}
Using~\eqref{eq:WE}, the weight enumerator of an optimal self-dual $[96,48,24]$ code is
\begin{align*}
& 1 + a y^{24}
+( 3082778880  - 24 a) y^{27}
+( 272857821696  + 276 a) y^{30}
\\ &
+( 18642386018880  - 2024 a) y^{33}
+( 827849897536896 + 10626 a) y^{36}
\\ &
+( 24804181974320640 - 42504 a) y^{39}
+( 505747055590698240 + 134596 a) y^{42}
\\ &
+ \cdots
+( - 13283136 + a) y^{96},
\end{align*}
where $a$ is an integer with $13283136 \le a \le 128449120$.
The Pless symmetry code $P_{96}$ of length $96$ is a self-dual
$[96,48,24]$ code (see~\cite{GNW}).
A double twistulant self-dual $[96,48,24]$ code can be found in~\cite[Table~3]{GG}.
The code $P_{96}$ has $A_{24}=15358848$~\cite{GNW}.
We verified by {\sc Magma}~\cite{Magma}
that the code in~\cite[Table~3]{GG} has $A_{24}=15358848$.
Our extensive search failed to discover a four-negacirculant
self-dual $[96,48,24]$ code.

\subsection{Self-dual $[100,50,21]$ codes}
From~\cite[Table~4]{GG}, the largest minimum weight
among self-dual codes of length $100$ is $21$, $24$ or $27$.
A double twistulant self-dual $[100,50,21]$ code can be found
in~\cite[Table~3]{GG}.
It was claimed that $C_{100}$ in~\cite[Table~VI]{GH08} is a self-dual $[100,50,21]$ code.
Unfortunately, $C_{100}$ in~\cite[Table~VI]{GH08} was incorrectly stated to be a four-circulant code.
The correct construction is four-negacirculant, that is,
the correct self-dual $[100,50,21]$ code $N_{100}$
is the four-negacirculant code with first rows $r_A$ and $r_B$ are
listed in~\cite[p.~417]{GH08} for
negacirculant matrices $A$ and $B$ in~\eqref{eq:G}.
We verified by {\sc Magma}~\cite{Magma} that the code in~\cite[Table~3]{GG}
has $A_{21}=14400$ and the code $N_{100}$ has $A_{21}=20900$.
This means that the code in~\cite[Table~3]{GG} performs better than $N_{100}$.

\bigskip
\noindent
{\bf Acknowledgment.}
This work was supported by JSPS KAKENHI Grant Number 15H03633.


\begin{landscape}
\begin{table}[thb]
\caption{Double circulant and double twistulant codes}
\label{Tab:Res}
\begin{center}
{\small
\begin{tabular}{c|c|c|c||c|c|c||c|c|c||c|c}
\noalign{\hrule height1pt}
$n$
& $d_{P}$  & $A_{d_{P}}$ & Code
& $d_{B}$  & $A_{d_{B}}$ & Code
& $d_{T}$  & $A_{d_{T}}$ & Code
& $d_{SD}$ & $A_{d_{SD}}$  \\
\hline
 4& 2&  2& $P_{4,1}, P_{4,2}$  & 2 & 2 & $B_4$  & 2 & 4 & $e_4$&  3& 8 \\
 8& 4& 20& $P_{8,1}$ in~\cite{DGH} & 4 & 22 & $B_{8}$ &4&24& $T_8$& 3& 16 \\
  & 3$^*$& 8 & $P'_8$ & 3$^*$ & 2 & $B'_8$ & 3$^*$ &8 &$T'_{8,1},T'_{8,2}$
 & &      \\
12& 5& 48& $P_{12,1},P_{12,2}$  & 5& 30& $B_{12}$ & 5 &8 &$T_{12}$ &  6&     264 \\
16& 6& 96& $P_{16,3}$ in~\cite{DGH}& 6 & 84 & $B_{16}$& 6 & 96 & $T_{16}$&  6&     224 \\
20& 7&200& $P_{20,3}$ in~\cite{DGH}& 7 &198 & $P_{20,8}$
  in~\cite{DGH}\footnotemark & 7 & 200 & $T_{20,1},T_{20,2},T_{20,3}$&  6&     120 \\
  &6$^*$ & 10 & $P'_{20}$ & 6$^*$ & 6 & $B'_{20}$  & 6$^*$  &20 &$T'_{20}$ & &   \\
24& 8&348& $P_{24}$  & 8& 264& $B_{24,1}B_{24,2},B_{24,3}$  & 8 & 312& $T_{24}$&  9&    4048 \\
28& 9&924& $P_{28}$  & 9& 832& $B_{28}$  & 9 & 616 & $T_{28}$&  9&    2184 \\
32&10&2208& $P_{32,1},P_{32,2}$ & 9 & 60 & $B_{32}$  & 9 & 32& $T_{32}$&  9&     960 \\
  & 9$^*$& 64& $P'_{32}$ & & &   &  &  & & &     \\
36&10&270&$P_{36}$ & 11 & 2244 & $B_{36}$  & 10 & 252 & $T_{36}$ & 12&   42840 \\
40&11& 720 &$P_{40}$ & 11 &722 &$B_{40}$ & 11 & 480 & $T_{40}$  & 12&   19760 \\
44&13 &19712 &$P_{44,1},P_{44,2}$ & 12 & 2436 &$B_{44}$
  &13 &19712 & $T_{44}$ & 12& 8008 \\
  &12$^*$ &1716 &$P'_{44}$ &  & &   &12$^*$ & 1716 & $T'_{44}$ & &  \\
\noalign{\hrule height1pt}
\end{tabular}
}
\end{center}
\end{table}
\footnotetext{The border values $(\alpha,\beta,\gamma)$ of $P_{20,8}$ were incorrectly reported in~\cite[Table 5]{DGH},
the correct values are $(2,1,1)$.}
\end{landscape}

\end{document}